\documentclass[10pt,pdfa]{article}
\usepackage[in]{fullpage}
\usepackage{url}
\usepackage{hyperref} 
\hypersetup{breaklinks=true}

\usepackage{mdframed}
\usepackage{dsfont}
\usepackage{tabularx}
\usepackage[normalem]{ulem}
\usepackage{comment}
\usepackage[shortlabels]{enumitem}
\usepackage{bm}
\usepackage{bbm}
\usepackage{amsfonts}
\usepackage{xspace}
\usepackage{amsmath}
\usepackage{amssymb}
\usepackage{amsthm}
\usepackage{xcolor}
\usepackage{graphicx}
\usepackage{qtree}
\usepackage{tree-dvips}
\usepackage{float}
\usepackage{physics}
\usepackage{authblk}
\usepackage{braket}
\usepackage[T1]{fontenc}
\usepackage{mathtools}
\usepackage{mathrsfs}
\usepackage{tikz}
\usepackage[linesnumbered,ruled,vlined]{algorithm2e}
\usepackage{qcircuit}
\usepackage[nameinlink,capitalize]{cleveref}

\newcommand{\subversion}[1]{}

  \hypersetup{colorlinks={true},linkcolor={blue},citecolor=magenta}

\newtheorem{theorem}{Theorem}[section]
\newtheorem*{theorem*}{Theorem}

\newtheorem{lemma}[theorem]{Lemma}
\newtheorem{claim}[theorem]{Claim}

\theoremstyle{remark}

\theoremstyle{definition}
\newtheorem{definition}[theorem]{Definition}

\numberwithin{equation}{section}
\newcommand\numberthis{\addtocounter{equation}{1}\tag{\theequation}}

\newtheorem*{rep@theorem}{\rep@title}
\newcommand{\newreptheorem}[2]{
\newenvironment{rep#1}[1]{
 \def\rep@title{#2 \ref{##1}}
 \begin{rep@theorem}}
 {\end{rep@theorem}}}
\makeatother

\newreptheorem{theorem}{Theorem}
\newreptheorem{lemma}{Lemma}
\newreptheorem{corollary}{Corollary}

\newcommand{\bit}{\{0,1\}}
\newcommand{\reg}[1]{\mathrm{#1}}
\newcommand{\ignore}[1]{}

\newcommand{\proj}[1]{\ketbra{#1}{#1}}

\newcommand{\E}{\mathop{\mathbb{E}}}
\newcommand{\N}{\mathbb{N}}

\newcommand{\C}{\mathbb{C}}
\newcommand{\negl}{\mathrm{negl}}

\newcommand\id{\mathbbm{1}}
\newcommand{\linear}{\mathrm{L}}
\newcommand{\unitary}{\mathrm{U}}

\newcommand{\bx}{\bm{x}}
\newcommand{\by}{\bm{y}}
\newcommand{\ba}{\bm{a}}
\newcommand{\ind}{\mathbbm{1}}
\newcommand{\supp}{\mathrm{supp}}

\newcommand{\ot}{\otimes}

\newcommand{\poly}{\mathrm{poly}}

\newcommand{\Haar}{\mathrm{Haar}}
\newcommand{\HaarMeasure}{\Haar}
\DeclarePairedDelimiterX{\hsip}[2]{\langle}{\rangle_{\tiny{\mathtt{HS}}}\xspace}{#1, #2}

\newcommand{\Keyspace}{\mathcal{K}}

\newcommand\algo{\mathcal}

\newcommand{\mom}[3]{\cM_{#1}^{(#2)}\left( #3 \right)}
\newcommand{\momhaart}[1]{\mom{{\rm Haar}}{t}{#1}}

\newcommand{\momr}[1]{\mom{R}{t}{#1}}

\newcommand{\mompft}[1]{\mom{PF}{t}{#1}}

\newcommand{\deq}{\coloneqq}

\newcommand{\distinct}{\mathrm{distinct}}

\newcommand{\cH}{\ensuremath{\mathcal{H}}}

\newcommand{\cK}{\ensuremath{\mathcal{K}}}

\newcommand{\cM}{\ensuremath{\mathcal{M}}}

\newcommand{\capprox}{\overset{c}{\approx}}

\title{Pseudorandom unitaries with non-adaptive security}

\author[1]{Tony Metger}
\author[2]{Alexander Poremba}
\author[3]{Makrand Sinha}
\author[4]{Henry Yuen}
\affil[1]{ETH Zurich}
\affil[2]{Massachusetts Institute of Technology}
\affil[3]{University of Illinois, Urbana-Champaign}
\affil[4]{Columbia University}

\begin{document}
\date{\vspace{-5ex}}

\maketitle

\begin{abstract}
Pseudorandom unitaries (PRUs) are ensembles of efficiently implementable unitary operators that cannot be distinguished from Haar random unitaries by any quantum polynomial-time algorithm with query access to the unitary. We present a simple PRU construction that is a concatenation of a random Clifford unitary, a pseudorandom binary phase operator, and a pseudorandom permutation operator.
We prove that this PRU construction is secure against non-adaptive distinguishers assuming the existence of quantum-secure one-way functions.
This means that no efficient quantum query algorithm that is allowed a single application of $U^{\ot \poly(n)}$ can distinguish whether an $n$-qubit unitary $U$ was drawn from the Haar measure or our PRU ensemble.
We conjecture that our PRU construction remains secure against adaptive distinguishers, i.e.~secure against distinguishers that can query the unitary polynomially many times in sequence, not just in parallel. 
\end{abstract}

\vspace{0.5cm}

\section{Introduction} \label{sec:intro}
Pseudorandom unitaries (PRUs) are ensembles\footnote{Strictly speaking, a PRU ensemble is an infinite sequence $\mathcal{U} = \{\mathcal{U}_n\}_{n \in \N}$ of $n$-qubit unitary ensembles $\mathcal{U}_n = \{U_k\}_{k \in \mathcal{K}}$, where $n \in \N$ serves as the security parameter (see \Cref{def:PRU} for a formal statement).} $\{U_k\}_{k \in \mathcal{K}}$ of unitaries that are efficient to implement, but that look indistinguishable from Haar random unitaries to any polynomial-time distinguisher.
This means that no polynomial-time distinguisher with oracle access to either a Haar random unitary or a unitary chosen uniformly from the PRU ensemble $\{U_k\}$ can tell the two cases apart.

PRUs are the natural quantum analogue to pseudorandom functions (PRFs), an idea that has proven enormously useful in classical computer science and cryptography. PRFs are efficient functions that look indistinguishable from uniformly random functions, just like PRUs are efficient unitaries that look indistinguishable from uniformly random (i.e.~Haar random) unitaries.
One reason why PRFs have been such a useful primitive in computer science is that they allow us to replace uniformly random functions, which are easy to analyse but cannot be implemented efficiently, by an efficient family of functions, as long as those functions are only queried by efficient algorithms.
Similarly, PRUs allow us to replace Haar random unitaries, which are well-understood but not efficiently implementable, by unitaries that are efficient to implement, as long as those unitaries are only applied by efficient quantum algorithms.\footnote{There is another way of replacing uniformly random functions by efficient functions: if we know ahead of time that a function will be queried at most $t$ times, we can replace a uniformly random function by a $t$-wise independent function family. Similarly, if we know that a unitary will only be applied at most $t$ times, we can replace a Haar random unitary by a $t$-design. In contrast, PRFs or PRUs are secure against \emph{all} polynomial time algorithms, i.e.~we do not need to know the number of queries the algorithm will make ahead of time.}

The concept of PRUs was introduced by Ji, Liu, and Song~\cite{ji2018pseudorandom}. 
Their paper gave a conjectured construction of PRUs, but only proved security (assuming quantum-secure one-way functions) for a much weaker primitive called \emph{pseudorandom states} (PRSs).
A pseudorandom state ensemble is a set of states (rather than unitaries) that look indistinguishable from Haar random states (even with access to polynomially many copies of the state).
Since their introduction, PRSs have become an influential concept with  applications in quantum cryptography~\cite{ananth2022cryptography,morimae2022quantum}, lower bounds in quantum learning theory~\cite{huang2022quantum}, and even connections to quantum gravity~\cite{aaronson2024quantum}.
However, proving security for a pseudorandom \emph{unitary} construction has remained an open problem and~\cite{haug2023pseudorandom} even shows restrictions on possible PRU constructions.

Given the difficulty of proving security for PRUs, recent work has considered intermediate steps between PRSs and PRUs.
\cite{lu2023quantum} introduced the notion of a pseudorandom state scrambler (PRSS).
A PRSS ensemble is a set of unitaries $\{U_k\}_{k \in \cK}$ such that for all states $\ket{\phi}$, the state family $\{U_k \ket{\phi}\}_{k \in \cK}$ is a PRS.
If we restrict the state $\ket{\phi}$ to the all-0 state $\ket{0}$, then we recover PRSs as a special case.
\cite{lu2023quantum} showed how to construct a PRSS ensemble (assuming quantum-secure one-way functions) by an intricate analysis of Kac's random walk.
In \cite{ananth2023pseudorandom}, the authors achieve a similar result to PRSS ensembles, except that they consider ensembles of isometries instead of unitaries.

One can also view PRSs and PRSSs as special cases of PRUs, where the distinguisher is restricted to particular queries: a PRS ensemble is a PRU ensemble that is secure against distinguishers that are only allowed a single query to $U^{\ot\poly(n)}$ on the state $\ket{0}^{\ot \poly(n)}$, and a PRSS is a PRU ensemble that is secure against distinguishers that are only allowed a single query to $U^{\ot\poly(n)}$ on a tensor power state $\ket{\phi}^{\ot \poly(n)}$ of their choice.

Our main result is a simple PRU construction with non-adaptive security.
This means that our PRU ensemble is secure against any distinguisher that makes polynomially many \emph{parallel} queries to the PRU, i.e.~the distinguisher is allowed a single query to $U^{\ot \poly(n)}$ on \emph{any input state}, rather than just the restricted classes of input states allowed in PRSs and PRSSs.
We conjecture that our construction also has adaptive security, i.e.~remains secure when the distinguisher is allowed polynomially many \emph{sequential} queries.

\paragraph{Our construction.}
Our PRU ensemble is a concatenation of a random Clifford unitary, a pseudorandom binary
phase operator, and a pseudorandom permutation operator\footnote{Somewhat interestingly, the only place imaginarity shows up (as shown necessary by~\cite{haug2023pseudorandom}) is in the Clifford unitary.}. More concretely, we require:
\begin{itemize}
    \item An ensemble of (quantum-secure) pseudorandom permutations (PRPs)~\cite{zhandry2016note}. Broadly speaking, this is a family $\{\pi_{k_1} : \{0,1\}^n \to \{0,1\}^n \}_{k_1 \in \mathcal{K}_1}$ of permutations with the property that, for a randomly chosen key $k_1 \sim \mathcal{K}_1$, the permutation $\pi_{k_1}$ is computationally indistinguishable from a perfectly random permutation.
    For a given $\pi_{k_1}$, we let $P_{k_1}$ be the corresponding $n$-qubit permutation matrix.
    \item An ensemble of (quantum-secure) pseudorandom functions (PRFs)~\cite{zhandry2021construct}. More formally, this is a family $\{f_{k_2} : \{0,1\}^n \to \{0,1\} \}_{k_2 \in \mathcal{K}_2}$ with the property that, for a randomly chosen key $k_2 \sim \mathcal{K}_2$, the function $\pi_{k_1}$ is computationally indistinguishable from a perfectly random Boolean function. For a given $f_{k_2}$, we let $F_{k_2}: \ket{x} \to (-1)^{f_{k_2}(x)} \ket{x}$ be the $n$-qubit phase oracle implementing $f_{k_2}$.
    \item An ensemble $\{C_{k_3} \}_{k_3 \in \mathcal{K}_{3}}$ of $n$-qubit Clifford unitaries. Note that this is an ensemble of size $2^{O(n^2)}$ from which we can efficiently sample~\cite{berg2021simple}. 
\end{itemize}
Then, our PRU is the $n$-qubit ensemble $\{U_k\}_{k \in \mathcal{K}}$ which is specified by a key $k = (k_1, k_2, k_3)$, where
\begin{equation}
    \label{eq:intro-construction}
    U_k = P_{k_1} F_{k_2} C_{k_3}.
\end{equation}
Note that, assuming the PRP and PRF scheme both have a key space consisting of $\mathcal{K}_1 = \mathcal{K}_2 = \bit^n$, our PRU ensemble has a key length of $|k| = n + n + O(n^2) = O(n^2)$, where $n \in \N$ is the security parameter. Because an ensemble of (quantum-secure) PRFs and PRPs can be constructed from the existence of (quantum-secure) one-way functions~\cite{zhandry2016note,zhandry2021construct}, the same assumption also suffices for our PRU ensemble.

\paragraph{Non-adaptive PRU security.} A non-adaptive quantum algorithm starts with an initial state $\ket{\psi}_{\reg{A}_1 \cdots \reg{A}_t \reg{B}}$ where registers $\reg{A}_1,\ldots,\reg{A}_t$ are each on $n$ qubits, and $\reg{B}$ is an arbitrary workspace register. The algorithm then applies the unitary $U^{\otimes t}$ on the registers $\reg{A}_1, \cdots, \reg{A}_t$ and performs a measurement afterwards. In order to show that no such algorithm can distinguish the PRU from a Haar random unitary, it suffices to show that for all $t = \poly(n)$ and all initial states $\ket{\psi}_{\reg{A}_1 \cdots \reg{A}_t \reg{B}}$, the following two density matrices are computationally indistinguishable:
\begin{equation}
    \label{eqn:security}
        \E_{k \sim \Keyspace} (U_k^{\ot t} \ot \id) \proj{\psi} (U_k^{\ot t} \ot \id)^\dagger \capprox \E_{U \sim \Haar} (U^{\ot t} \ot \id) \proj{\psi} (U^{\ot t} \ot \id)^\dagger~.
\end{equation}
Note that in the above, the registers $\reg{A}_1,\ldots,\reg{A}_t$ of $\ket{\psi}$ may all be entangled with each other, the reduced states on each register may be different from each other, and finally the distinguisher is allowed access to the purification of the input state to the unitaries. In contrast, as mentioned before, the security of PRSS was only established for identical pure-state inputs to the queries. 

\begin{theorem}
    Assuming the existence of quantum-secure one-way functions, the ensemble described in~\eqref{eq:intro-construction} satisfies non-adaptive PRU security. 
\end{theorem}

\subsection{Proof overview} 
For this proof overview, we assume that our algorithm has no workspace, i.e. the initial state is $\ket{\psi}_{\reg{A}_1\cdots\reg{A}_t}$. This is merely to simplify the notation, the proof works exactly in the same way even in the case of an entangled workspace. In order to establish computational indistinguishability as in \cref{eqn:security}, we replace the pseudorandom permutation and function with their truly random counterparts and show information-theoretic indistinguishability. Let $P$ denote a uniformly random permutation matrix on $n$-qubits, let $F: \ket{x} \to (-1)^{f(x)}\ket{x}$ be a diagonal unitary with a uniformly random function $f$, and let $C$ be a uniformly random $n$-qubit Clifford, all sampled independently. It then suffices to show that the following two density matrices are close in trace distance:
\begin{equation}
    \label{eqn:inf-security}
      \E_{PFC} (PFC)^{\ot t} \proj{\psi} ((PFC)^{\ot t})^\dagger \approx \E_{U \sim \Haar} U^{\ot t}  \proj{\psi} (U^{\ot t})^\dagger~.
\end{equation}

For this, we leverage Schur-Weyl duality to compute what these two density matrices explicitly look like. Let $d=2^n$ be the dimension.
According to Schur-Weyl duality:
\begin{enumerate}[(1)]
    \item the space $(\C^d)^{\ot t}$ can be decomposed as a direct sum\footnote{Here $\lambda \vdash t$ means that $\lambda$ is a partition (commonly described by a Young diagram) of $[t]$.} $\bigoplus_{\lambda \vdash t} P_\lambda$ where $P_\lambda = {W_\lambda \ot V_\lambda}$ is a tensor product of two spaces, and
    \item any unitary $U^{\otimes t}$ only acts non-trivially on the subspaces $W_{\lambda}$ and any unitary $R_\pi$ that permutes the $t$ subsystems according to a permutation $\pi \in S_t$ acts non-trivially only on the subspaces $V_{\lambda}$.
\end{enumerate}

Using this, we show that applying a \emph{$t$-wise Haar twirl} to $\ket{\psi}$ to obtain the state on the right hand side of \cref{eqn:inf-security} results in the following state:
\begin{equation*}
    \E_{U \sim \Haar} U^{\ot t}  \proj{\psi} (U^{\ot t})^\dagger = \sum_{\lambda \vdash t} \frac{\ind_{W_\lambda}}{\Tr[\ind_{W_\lambda}]} \otimes \Tr_{W_\lambda}[\ind_{P_\lambda} \ketbra{\psi}{\psi}\ind_{P_\lambda}]\,,
\end{equation*}
where $\ind_{W_{\lambda}}$ is the identity on the subspace $W_\lambda$, i.e.~the orthogonal projection onto that subspace.
In particular, the state is a direct sum of tensor product states where the state on the subspaces $W_{\lambda}$ is maximally mixed. 

Next, we would like to show that applying a \emph{$t$-wise $PFC$ twirl} to $\ket{\psi}$ to obtain the left hand side state in \cref{eqn:inf-security} results in a state that is close to the above. For technical reasons, the computations here are easier if the state $\ket{\psi}$ is supported only on the subspace of distinct computational basis states in the registers $\reg{A}_1, \ldots \reg{A}_t$, i.e., the subspace spanned by $\ket{x_1,\ldots,x_t}$ where $x_1,\ldots, x_t$ are all distinct. We show that applying a random Clifford ensures this since
\[\Tr[\Lambda \E_{C} C^{\ot t} \ketbra{\psi}{\psi} C^{\ot t,\dagger}] \geq 1 - O(t^2/d),\]
where $\Lambda$ denotes the projector on the \emph{distinct subspace}.

For states $\ket{\psi}$ in the distinct subspace, we show that a $t$-wise $PF$ twirl results in
\begin{equation*}
    \E_{PF} (PF)^{\ot t}  \proj{\psi} ((PF)^{\ot t})^\dagger = \sum_{\lambda \vdash t} \frac{\ind_{\Lambda_\lambda}}{\Tr[\ind_{\Lambda_{\lambda}}]} \otimes \Tr_{W_\lambda}[\ind_{P_\lambda} \ketbra{\psi}{\psi}\ind_{P_\lambda}],
\end{equation*}
where $\Lambda_\lambda$ is a subspace of $W_{\lambda}$ that comes from decomposing the distinct subspace projector in terms of the Schur-Weyl subspaces. 
We show that $\Lambda_{\lambda}$ fills most of $W_\lambda$, i.e.~the dimension of the subspace $\Lambda_\lambda$ is close to the dimension of $W_\lambda$:
$$\frac{\Tr[\ind_{\Lambda_{\lambda}}]}{\Tr[\id_{W_\lambda}]} = 1- O\left(\frac{t^2}{d}\right).$$
This implies that the mixed state on $\Lambda_\lambda$ is close in trace distance to the maximally mixed state on $W_{\lambda}$: 
$$\left\|\frac{\ind_{\Lambda_\lambda}}{\Tr[\ind_{\Lambda_{\lambda}}]} - \frac{\ind_{W_\lambda}}{\Tr[\ind_{W_\lambda}]}\right\|_1 = O \Big(\frac{t^2}d \Big)\,.$$

Since this is true for each $\lambda$, putting all the above together, we get that left and right hand sides in \cref{eqn:inf-security} have trace distance at most $t^2/d$, which is exponentially small since $t= \poly(n)$ and $d=2^n$. This completes the sketch of the proof.

\subsection{Discussion and future directions}

We prove that our PRU construction is secure against non-adaptive adversaries.
There are two directions in which one would like to strengthen this result: allow the adversary to access the inverse of the unitary, too, and allowing the adversary to query the unitary adaptively, i.e.~make a query, then apply some efficient unitary to the resulting state, make another query on that state, and so on.

Aside from constructing PRUs, their applications are also largely unexplored.
Given the utility of PRFs in classical computer science, PRUs appear to be a fundamental primitive for quantum computer science, but not much is known about concrete applications.

One natural area of application is quantum cryptography.
For example,~\cite{lu2023quantum} showed that PRUs (and even PRSSs) can be useful for multi-copy quantum cryptography:
in most encryption schemes for quantum messages (e.g., the quantum one-time pad and its variants), if we wish to encrypt multiple copies of the same quantum state, we need to sample fresh keys for each new ciphertext.~\cite{lu2023quantum} observed that Haar (pseudo)randomness allows one to perform such a task in a compact manner with only a single key (for an arbitrary polynomial amount of copies). If the state to be encrypted is guaranteed to be unentangled with the environment, the PRSSs from~\cite{lu2023quantum} suffice; in the general case which allows for entanglement with an auxiliary system, non-adaptively PRUs as constructed in our work seem necessary. 
As another example, PRUs might be useful in the context of unclonable cryptography. Many constructions, such as those for unclonable encryption~\cite{broadbent2020uncloneable} or quantum copy-protection~\cite{coladangelo2020quantum,coladangelo2021hidden}, make use of either Wiesner states or subspace coset states---both of which are completely broken once identical copies become available. It seems plausible that one could use PRUs to construct multi-copy secure unclonable encryption schemes, and even multi-copy secure quantum copy-protection schemes.

Another area of application concerns the time evolution of chaotic quantum systems. In the past few years, a series of works has proposed using Haar random unitaries as ``perfect scramblers''~\cite{Hayden_2007,brown2012scrambling,susskind2016computational} to model such dynamics. However, a more recent line of work has instead shifted towards quantum pseudorandomness in order to model such phenomena in terms of \emph{efficient} processes.
For example, Kim and Preskill~\cite{Kim_2023} use PRUs to model the internal dynamics of a black hole, whereas Engelhardt et al.~\cite{engelhardt2024cryptographic} use PRUs to model the time evolution operator of a holographic conformal field theory. Due to their scrambling properties, it is conceivable that PRUs will find more applications in theoretical physics.

\paragraph{Acknowledgments.} We thank Prabhanjan Ananth, Adam Bouland, Tudor Giurgica-Tiron, Jonas Haferkamp, Isaac Kim, and John Wright for helpful discussions. We also thank John for suggesting a simplified proof of \Cref{lem:K_simplified}.
We thank the Simons Institute for the Theory of Computing, where some of this work was conducted.
TM acknowledges support from the ETH Zurich Quantum Center, the SNSF QuantERA project (grant 20QT21\_187724), the AFOSR grant FA9550-19-1-0202, and an ETH Doc.Mobility Fellowship. AP is supported by the National Science Foundation (NSF) under Grant No. CCF-1729369.
HY is supported by AFOSR awards
FA9550-21-1-0040 and FA9550-23-1-0363, NSF CAREER award CCF-2144219, and the Sloan Foundation.

\section{Preliminaries}

\subsection{Notation and basic definitions}
\label{sec:notation}
\paragraph{Linear Maps.} For a Hilbert space $\cH$, we denote by $\linear(\cH)$ linear operators on $\cH$.
A map $\cM: \linear(\cH) \to \linear(\cH)$ is called a quantum channel if it is completely positive and trace-preserving.
If $\cH'$ is an additional Hilbert space and $O \in \linear(\cH \ot \cH')$ is an operator on a larger space, we write $\cM(O)$ to mean $\cM$ applied to the $\cH$-subsystem, i.e.~$\cM(O)$ is shorthand for $(\cM \ot \id_{\cH'})(O)$. For a subspace $W$ of a vector space $V$, we denote by $\ind_{W}$ the orthogonal projection onto $W$. 

\paragraph{Distinct tuples.} We use bold faced fonts to denote tuples. For a tuple $\bx = (x_1,\dots,x_t) \in [d]^t$ and a permutation $\sigma \in S_t$, we write $\bx_{\sigma} = (x_{\sigma(1)}, \cdots,x_{\sigma(t)})$ for the tuple where the indices are permuted according to $\sigma.$ We call a tuple $\bx \in [d]^t$ distinct if $x_i \neq x_j$ for all $i \neq j$. We denote the set of distinct tuples in $[d]^t$ by $\distinct(d,t)$.
We also define the projector onto the subspace of distinct tuples 
\begin{equation}\label{def:distinct}
    \ \Lambda_{d,t} = \sum_{\substack{\bx \in \distinct(d,t)}} \proj{\bx} \,.
\end{equation}
We will frequently drop the indices $d,t$ and just write $\Lambda$, with $d$ and $t$ being clear from context.

\paragraph{Permutation  operators.} \label{sec:prelim_phase_perm}

We define the following permutation operator on $\C^d$.

\begin{definition}[Permutation operator on $\C^d$]
Define the permutation operator $P_\pi$ on $\C^d$ for $\pi \in S_d$ to be the linear map
\begin{align}\label{def:permutationop}
P_\pi: \ket{x} \mapsto \ket{\pi(x)}\,.
\end{align}
\end{definition}

We will frequently consider uniformly random permutation operators on $\C^d$. We will suppress the dependence on $\pi$ and write the random operator as $P$.

\paragraph{Symmetric group and representations.}

Unitary representations of a group allow us to represent the elements of the group as unitary matrices over a vector space in a way that the group operation is represented by matrix multiplication. We consider the following representation of the symmetric group which permutes the tensor factors.

\begin{lemma}[Representation of $S_t$ on tensor product spaces]
For any fixed $d$, define the permutation operator $R_\pi$ on $(\C^d)^{\ot t}$ for $\pi \in S_t$ to be the map
\[
    R_\pi: \ket{\ba} \mapsto \ket{\ba_{\pi^{-1}}}~.
\]
Then $((\C^d)^{\ot t}, R_\pi$ forms a unitary representation of $S_t$.
Note that we leave the dependence of $R_\pi$ on the choice of $d$ implicit.
\end{lemma}

We note that the representation $((\C^d)^{\ot t}, R_{(\cdot)})$ can be decomposed into (isotypic) copies of the irreducible representations (or irreps) of the symmetric group $S_t$, which we denote by $\{(V_{\lambda}, R^{\lambda}_{(\cdot)})\}_{\lambda}$, where $\lambda \vdash t$ is a partition of $t$ (or Young diagram with at most $t$ boxes) and $V_\lambda$ are vector spaces called Specht modules.

\paragraph{Binary phase operators.} For a function $f: [d] \to \bit$, we define the binary phase operator 
\begin{align}\label{def:binaryphaseop}
F_f: \ket{x} \mapsto (-1)^{f(x)} \ket{x}.
\end{align} 
We will frequently consider uniformly random binary phase operators, which we will just write as $F$.

\paragraph{Haar measure and $t$-wise twirling.} We recall the definition of the Haar measure and $t$-wise twirl.

\begin{definition}[Haar measure]
The Haar measure is the unique left- and right-invariant probability measure on the unitary group $\unitary(d)$.
Throughout this paper, we denote sampling from the Haar measure over $\unitary(d)$ by $U \sim \HaarMeasure(d)$. If the dimension $d$ is clear from the context, we simply write $U \sim \HaarMeasure$.
\end{definition}

\begin{definition}[$t$-wise $R$-twirl] \label{def:twirl}
Let $R$ be a random unitary matrix sampled from some probability measure over the unitary group $\mathrm{U}(d)$.  We define the $t$-wise $R$-twirl as the following operator 
\begin{align*}
\momr{\cdot} = \E_{R} R^{\ot t} (\cdot) R^{\ot t, \dagger} \,.
\end{align*}
\end{definition}

When $R$ is sampled from the Haar measure over $\mathrm{U}(d)$, we refer to it as $t$-wise Haar twirl and denote it by $\momhaart{\cdot}$. We will also frequently consider the $t$-wise $\mompft{\cdot}$ twirl, where $R = PF$ is a product of a random permutation operator $P$ and a binary phase operator $F$, which are sampled independently.
We call this the permutation-phase twirl.

\subsection{Schur-Weyl duality}
Consider the representation $R_\pi$  of the symmetric group $S_t$ and the representation $U^{\ot t}$ of the unitary group $\mathrm{U}(d)$ on the vector space $(\C^d)^{\ot t}$. Schur-Weyl duality says that the irreducible sub-representations of these two representation fit together nicely. 
\begin{lemma}[{Schur-Weyl duality, see e.g.~\cite[Theorem 1.10]{christandl2006structure}}] \label{lem:schur-weyl}
The tensor product space $(\C^d)^{\ot t}$ can be decomposed as
\[
    (\C^d)^{\ot t} \cong \bigoplus_{\lambda \vdash t} P_{\lambda}~ \text{ with } P_{\lambda} = W_{\lambda} \ot V_{\lambda}
\]
where $\lambda \vdash t$ indexes partitions of $\{1, \dots, t\}$, which are commonly represented by Young diagrams.\footnote{We note that throughout this paper $d \gg t$, otherwise the Young diagrams need to be restricted to $d$ rows.}

The \emph{Weyl modules} $W_\lambda$ are irreducible subspaces for the unitary group $\mathrm{U}(d)$ and the \emph{Specht modules} $V_\lambda$ are irreducible subspaces for the symmetric group $S_t$. Consequently, the action of the product group $S_t \times U_d$ on $(\C^d)^{\ot t}$ decomposes as
\begin{align*}
    R_\pi  = \sum_{\lambda \vdash t} \ind_{W_{\lambda}} \ot R_\pi^{(\lambda)} \text{\hphantom{t} and \hphantom{t}} &  U^{\ot t}  = \sum_{\lambda \vdash t} U^{(\lambda)} \ot \ind_{V_{\lambda}}~,
\end{align*}
where $(W_\lambda, U^{(\lambda)})$ and $(V_\lambda, R_\pi^{(\lambda)})$ are irreducible representations of the unitary group $\unitary(d)$ and the symmetric group $S_t$, respectively. 
\end{lemma}

We denote the specific basis which block-diagonalizes all the above operations as the Schur-Weyl basis.

\begin{definition}[Schur-Weyl basis]
Let $\{\ket{w_{\lambda, i}}\}_{i}$ and $\{\ket{v_{\lambda, j}}\}_j$ be orthonormal bases of $W_\lambda$ and $V_\lambda$, respectively.
Then we call 
\begin{align*}
\{\ket{w_{\lambda, i}} \ot \ket{v_{\lambda, j}}\}_{\lambda, i, j}
\end{align*}
a Schur-Weyl basis of $(\C^d)^{\ot t}$ (where we interpret each vector $\ket{w_{\lambda, i}} \ot \ket{v_{\lambda, j}} \in P_\lambda$ as a $(\C^d)^{\ot t}$-vector in the natural way).
\end{definition}

The following decomposition of the distinct subspace projector $\Lambda$ easily follows from Schur-Weyl duality: since $\Lambda$ is invariant under permutation of the tensor factors (i.e.~$R_{\pi} \Lambda = \Lambda R_{\pi}$ for all $\pi \in S_t$), it acts as an identity on the Specht modules $V_{\lambda}$ by Schur's lemma. The same decomposition also holds for any permutation-invariant operator but we shall only need the following.

\begin{lemma}[Decomposition of the distinct subspace projector] \label{lem:perm_inv_sw}
Let $\Lambda \in \linear((\C^d)^{\ot t})$ be the projector on the tuples of distinct strings defined in \Cref{def:distinct}. 
Then, 
\begin{align*}
\Lambda = \sum_{\lambda \vdash t} \Lambda^{(\lambda)}_{W_\lambda} \ot \id_{V_\lambda} \,,
\end{align*}
with each $\Lambda^{(\lambda)}_{W_{\lambda}}$ a projector on a subspace of $W_\lambda$. 
\end{lemma}

We will also need the following relation between the dimensions of the Weyl and Specht modules.
\begin{lemma}[\cite{christandl2006structure}, Theorem 1.16] \label{lem:dim_weyl_specht}
The dimensions of the Weyl and Specht modules $W_\lambda$ and $V_\lambda$ satisfy
\[
    \dim(W_\lambda) = \frac{\dim(V_\lambda)}{t!} \prod_{(i,j) \in \lambda} (d + j - i) \,,
\]
where $(i,j)$ denotes the row and column number of a box in the Young diagram corresponding to $\lambda$. 
\end{lemma}

The following lemma follows from the standard formula for computing the projector on isotypic copies of an irreducible subspace (see~\cite[Section 2.4]{fulton2013representation}). Here we apply it to the representation of the symmetric group $S_t$  over $(\C^d)^{\ot t}$, where Schur-Weyl duality implies that the subspace of all isotypic copies of the Specht modules $V_{\lambda}$ is exactly $P_\lambda = W_\lambda \ot V_\lambda$.

\begin{lemma}\label{lem:P_characters}
The projection onto the subspace $P_\lambda$ is given by 
\begin{equation}
\label{eq:isotypical-projection}
    \ind_{P_\lambda} = \frac{\dim(V_\lambda)}{t!} \sum_{\pi \in S_t} \chi_\lambda(\pi^{-1}) R_\pi \,,
\end{equation}
where $\chi_\lambda(\cdot) = \Tr[R^{\lambda}_{(\cdot)}]$ is the character corresponding to the irrep $(V_{\lambda}, R^{(\lambda)}_{(\cdot)})$.
\end{lemma}

We will also need (a special case of) the standard Schur orthogonality relations for matrix coefficients (see~\cite{Bump}, Theorems 2.3 and 2.4). These relations say that if we express unitary irreducible representations of a group in any basis, then the different matrix entries are orthogonal under an inner product obtained by averaging over the group. Here we specialize the above to the Schur-Weyl basis and the unitary irreducible sub-representations $R^{\lambda}_{(\cdot)}$ of $R_{(\cdot)}$ as given in \Cref{lem:schur-weyl}.
\begin{lemma}[Schur orthogonality relations] \label{lem:orthogonality}
    Let $(V_{\lambda}, R^{\lambda}_{(\cdot)}), (V_{\lambda'}, R^{\lambda'}_{(\cdot)})$ be two irreducible representations of the symmetric group $S_t$. Then,
    \[ \E_{\pi \in S_t}\left[ \bra{v_{\lambda, i}} R^{\lambda}_{\pi}  \ket{v_{\lambda, j}}  \overline{\bra{v_{\lambda', k}} R^{\lambda'}_{\pi}  \ket{v_{\lambda', \ell}}} \right] = \frac{1}{\dim(V_{\lambda})} \delta_{\lambda, \lambda'} \delta_{i,k} \delta_{j,\ell}\,. \]
\end{lemma}

\section{Proof of non-adaptive security}
\label{sec:def-construction}

In this section, we first give a formal definition of PRUs, as proposed by Ji, Liu, and Song~\cite{ji2018pseudorandom}. Then, we prove that our construction in \Cref{eq:intro-construction} is a non-adaptive pseudorandom unitary.

\begin{definition}[Pseudorandom unitary]\label{def:PRU} Let $n \in \N$ be the security parameter. An infinite sequence $\mathcal{U} = \{\mathcal{U}_n\}_{n \in \N}$ of $n$-qubit unitary ensembles $\mathcal{U}_n = \{U_k\}_{k \in \mathcal{K}}$ is a pseudorandom unitary if it satisfies the following conditions.
\begin{itemize}
    \item (Efficient computation) There exists a polynomial-time quantum algorithm $\mathcal{Q}$ such that for all keys $k \in \mathcal{K}$, where $\mathcal{K}$ denotes the key space, and any $\ket{\psi} \in (\C^2)^{\ot n}$, it holds that
    $$
    \mathcal{Q}(k,\ket{\psi}) = U_k \ket{\psi}\,.
    $$

    \item (Pseudorandomness) The unitary $U_k$, for a random key $k \sim \algo K$, is computationally indistinguishable from a Haar random unitary $U \sim \HaarMeasure(2^n)$. In other words, for any quantum polynomial-time (QPT) algorithm $\algo A$, it holds that $$
    \vline\, \underset{k \sim \algo K}{\Pr}[\algo A^{U_k}(1^\lambda)=1] - \underset{U \sim \Haar}{\Pr}[\algo A^{U}(1^\lambda) =1]  \,\vline \,\leq \, \negl(n)\,.
    $$ 
We call $\mathcal{U} = \{\mathcal{U}_n\}_{n \in \N}$ a non-adaptive pseudorandom unitary if $\algo A$ is only allowed to make parallel queries to the unitary $U_k$ (or $U$ in the Haar random case).
\end{itemize}
Note that, whenever we write $\mathcal{U}_n = \{U_k\}_{k \in \mathcal{K}}$, it is implicit that the key space $\mathcal{K}$ depends on the security parameter $n \in \N$, and that the length of each key $k \in \mathcal{K}$ is polynomial in $n$.
\end{definition}

Our main result is that the construction in \Cref{eq:intro-construction} is indeed a non-adaptive PRU.

\begin{theorem} 
Let $n \in \N$ be the security parameter. Then, the ensemble
$\mathcal{U}_n = \{U_k\}_{k \in \mathcal{K}}$ of $n$-qubit unitary operators defined in \Cref{eq:intro-construction} is a non-adaptive pseudorandom unitary when instantiated with ensembles of $n$-bit (quantum-secure) PRFs and PRPs.
\end{theorem}
\begin{proof} 
From the construction, it is clear that a random unitary $U_k$ from the above family can be sampled efficiently (see e.g.~\cite{berg2021simple} for simple way to sample a uniform Clifford unitary).

We now show the non-adaptive security of the above PRU family. A non-adaptive quantum adversary starts with an initial state $\ket{\psi}_{\reg{A} \reg{E}}$, where register  $\reg A \cong ((\C^2)^{\ot n})^{\ot t}$ is on $nt$ qubits, and $\reg{E}$ is an arbitrary workspace register. The algorithm then applies the unitary $U_k^{\otimes t}$ on the register $A$ and performs a measurement afterwards. For the rest of the proof, we set $d=2^n$ to be the dimension of each input register. In order to argue security, it suffices to show that the density matrices
\[ \rho:= \E_{k \in \cK} (U_k)^{\ot t}_{\reg A} \proj{\psi}_{\reg{AE}} (U_k)^{\ot t, \dagger}_{\reg A} \text{\hphantom{t} and } \rho^{\rm hr}: = \E_{U \sim \Haar} U^{\ot t}_{\reg A} \proj{\psi}_{\reg{AE}} U^{\ot t, \dagger}_{\reg A}, \]
are computationally indistinguishable with at most negligible advantage.

From the post-quantum security of the PRF and PRP families used in \Cref{eq:intro-construction}, it follows immediately that if we replace the pseudorandom permutation and function with their fully random counterparts, then the states are computationally indistinguishable up to negligible advantage in $n$. Consequently, let $P$ denote a uniformly random permutation operator on $n$-qubits as defined in \Cref{def:permutationop}, let $F$ be a uniformly random binary phase operator as in \Cref{def:binaryphaseop} and let $C$ be a uniformly random $n$-qubit Clifford, all sampled independently. Then,
it suffices to argue that the following ``fully random'' density matrix 
\begin{align*}
\rho^{\rm fr} \deq \E_{F, P, C} (P F C)^{\ot t}_{\reg A} \proj{\psi}_{\reg{AE}} (P F C)^{\ot t, \dagger}_{\reg A} 
\end{align*}
is close in trace distance to $\rho^{\rm hr}$. Recalling the definition of $t$-wise $R$-twirl operator (acting on the register $\reg{A}$, which we shall omit from the notation henceforth), our goal is to bound the following trace distance:
\begin{align}\label{eqn:td}
    \ \|\rho^{\rm fr} - \rho^{\rm hr}\|_1 = \left\|\mompft{\underbrace{\E_{C} C^{\ot t}_{\reg A} \proj{\psi}_{\reg{AE}} C^{\ot t, \dagger}_{\reg A} }_{:=\xi_{\reg{AE}}}} - \momhaart{ \E_{C} C^{\ot t}_{\reg A} \proj{\psi}_{\reg{AE}} C^{\ot t, \dagger}_{\reg A} }\right\|_1 \,,
\end{align}
where the equality uses the fact that the product unitary $UC$ is also Haar distributed by the invariance of the Haar measure.

\paragraph{Reduction to distinct subspace.} To show the above, the calculations are easier if we restrict attention to the subspace of $(\C^d)^{\ot t}$, that consists of distinct strings, i.e.~basis states of the form $\ket{\bx}$ where $\bx = (x_1,\dots,x_t) \in [d]^t$ is a tuple of distinct strings. We show in \Cref{sec:clifford} that applying a $t$-wise Clifford twirl ensures that the input state has a large overlap with this subspace. 

\begin{lemma}[Clifford twirl and distinct subspace] \label{lem:clifford}
Let $\Lambda$ be the projector on the distinct subspace defined in \Cref{def:distinct}. Then, for any state $\rho$ on the register $\reg{A}$, we have
    $$\Tr[\Lambda \E_{C} {{C}^{\ot t}} \rho \, {{C}^{\ot t,\dagger}}] \geq 1 - O(t^2/d)\,.$$
\end{lemma}

Let $\phi_{\reg{AE}}$ be the mixed state obtained by normalizing the (positive semi-definite) matrix $\Lambda_{\reg{A}} \xi_{\reg{AE}} \Lambda_{\reg{A}}$, where $\xi_{\reg{AE}}$ is defined in \Cref{eqn:td}. Then the above together with the gentle measurement lemma implies that $\|\phi - \xi\|_1 \le O(t/\sqrt{d})$. Consider a purification $\ket{\phi}_{\reg{A \tilde E}}$ of $\phi_{\reg{AE}}$ that satisfies $\Lambda_{\reg A} \ket{\phi}_{\reg{A \tilde E}} = \ket{\phi}_{\reg{A \tilde E}}$; such a purification exists by construction of $\phi_{\reg{AE}}$. Then, 
\begin{align*}
\norm{\rho^{\rm fr} - \rho^{\rm hr}}_1 
&\leq \norm{\mompft{\phi_{\reg{A E}}} - \momhaart{\phi_{\reg{A E}}}}_1 + O(t/\sqrt{d})\\
&\leq \norm{\mompft{\proj{\phi}_{\reg{A \tilde E}}} - \momhaart{\proj{\phi}_{\reg{A \tilde E}}}}_1 + O(t/\sqrt{d}) \,, \numberthis \label{eqn:haarvspf}
\end{align*}
where we use that a $t$-wise twirl is a quantum channel and the 1-norm can only decrease under partial trace. 
Thus, we may assume for the rest of the the proof that the input is supported over the distinct subspace of register $\reg{A}$.

\paragraph{Haar twirl vs permutation-phase twirl.} In \Cref{eqn:haarvspf}, the twirling operators only act on register $\reg A$, and since the input state $\ket{\phi}_{\reg{A \tilde E}}$ is only supported on the distinct string subspace of $A$, it remains to show that on this subspace of distinct strings, the Haar twirl and $PF$-twirl essentially act similarly.

\begin{lemma} \label{lem:haar_pf_close}
Let $\reg{A} \cong (\C^d)^{\ot t}$, let $\reg{\tilde E}$ be an arbitrary quantum register, and let $\ket{\phi}_{\reg{A \tilde E}}$ be a state that satisfies $\Lambda_{\reg A} \ket{\phi}_{\reg{AE}} = \ket{\phi}_{\reg{AE}}$.
Then, 
\begin{align*}
\norm{\momhaart{\proj{\phi}_{\reg{A\tilde E}}} - \mompft{\proj{\phi}_{\reg{A \tilde E}}}}_1 \leq O(t^2/d) \,.
\end{align*}
\end{lemma}

The proof of the above relies on computing what the two states look like in the Schur-Weyl basis and is presented in \Cref{sec:twirl}. 

Putting everything together, we have that 
\[ \|\rho^{\rm fr} - \rho^{\rm hr}\|_1 \le O(t/\sqrt{d})\,.\]

Since $t$ is polynomial in $n$ and $d=2^n$, the non-adaptive security follows. \qedhere
\end{proof}

\subsection{Clifford twirl and distinct subspace (proof of \Cref{lem:clifford})}
\label{sec:clifford}

We will omit the registers $\reg{A} = \reg{A_1 \dots A_t}$ from the notation unless needed. Our goal is to show that applying a $t$-wise Clifford twirl on any input state $\rho$ produces a state that has a large overlap with the distinct subspace. The proof only relies on the standard fact that the Clifford group forms a $2$-design. We consider the projector on the orthogonal complement of the distinct subspace
    \begin{align*}
    \bar \Lambda = \id - \Lambda = \sum_{\bx \in [d]^t \setminus \distinct(d,t)} \proj{\bx} \,,
    \end{align*}
and decompose the projector into $O(t^2)$ sub-projectors according to which elements collide: since any tuple of non-distinct strings must have at least two equal entries, we have
\begin{align*}
    \bar \Lambda \preccurlyeq \sum_{1\le i < j \le t} \Pi_{ij} \otimes \id_{[n]\setminus \{ij\}}\,, \quad \text{ where } \Pi_{ij} = \sum_{x \in [d]} \proj{x}_i \ot \proj{x}_j, 
    \end{align*}
where the subscript $i$ denotes the register $\reg{A}_i$. We shall omit the identity from the notation below.

Using cyclicity of trace, we have
    \begin{align*}
    \Tr[\bar \Lambda \E_{C} C^{\ot t} \rho C^{\ot t}] 
    &\leq \sum_{i < j} \Tr[\E_{C} C^{\ot t, \dagger} \Pi_{ij}  C^{\ot t} \rho] 
    = \sum_{i < j} \Tr[\E_{C} (C_i^\dagger \ot C_j^{\dagger})  \Pi_{ij} (C_i \ot C_j) \rho],  
\end{align*}
where for the second equality, we cancelled the $C$-unitaries on all systems except $i$ and $j$, with $C_i$ denoting application of $C$ on system $i$.

Using the standard fact that the Clifford group forms a 2-design, we can replace the average over Clifford unitaries with an average over the Haar measure in the above expression, which only uses the second moment. Thus,
\begin{align*}
    &   \Tr[\bar \Lambda \E_{C} C^{\ot t} \rho C^{\ot t}]  \le \sum_{i < j} \Tr[\E_{U \sim \HaarMeasure} (U_i^\dagger \ot U_j^{\dagger}) \Pi_{ij}  (U_i \ot U_j) \rho] = \sum_{i < j} \Tr[\E_{U \sim \HaarMeasure} (U^\dagger \ot U^{\dagger}) \Pi_{ij} (U \ot U) \rho_{ij}],
\end{align*}
where for the last equality, we performed the partial trace over all systems except $i$ and $j$, with $\rho_{ij}$ denoting the reduced state on these systems. Since $\rho_{ij}$ is a quantum state, we can bound each of the $O(t^2)$ trace terms by the corresponding operator norms to obtain
    \begin{align}\label{eqn:2}
    \tr[\bar \Lambda \E_{C} C^{\ot t} \rho C^{\ot t}] 
    &\leq O(t^2d) \norm{\E_{U \sim \HaarMeasure} (U \ot U)^{\dagger} \left(\frac{\Pi}{\Tr[\Pi]}\right) (U \ot U)}_\infty \,,
    \end{align}
where $\Pi = \sum_{x\in [d]} \proj{x} \ot \proj{x}$ with $\Tr[\Pi]=d$. 

Since applying a Haar random unitary on any state gives a Haar random state $\ket{\psi} \in \C^d$, by linearity of expectation the expression inside the operator norm is 
\[ \E_{\ket{\psi} \sim \HaarMeasure}[\proj{\psi} \ot \proj{\psi}]. \]
It is well known that this is the maximally mixed state on the symmetric subspace $\text{Sym}^2(d)$ of $\C^d \ot \C^d$, which has dimension  $\frac{d(d+1)}{2}$ (see \cite{harrow2013church}, Proposition 6). Thus, the operator norm is $\frac{2}{d(d+1)}$ and inserting this into \Cref{eqn:2} yields the claimed result.

\subsection{Haar twirl vs permutation-phase twirl (proof of \Cref{lem:haar_pf_close})} \label{sec:twirl}

In order to compare the action of the Haar and the permutation-phase twirl, we express the result of applying these operators in the Schur-Weyl basis. In particular, we consider states $\ket{\phi}_{\reg{A \tilde E}}$ where the register $\reg{A}$ is supported over the distinct subspace, i.e.~states satisfying $\Lambda_{\reg A} \ket{\phi}_{\reg{A\tilde E}} = \ket{\phi}_{\reg{A\tilde E}}$.
To apply Schur-Weyl duality, we decompose $\reg{A} \cong (\mathbb{C}^d)^{\ot t} \cong \bigoplus_{\lambda \vdash t} P_\lambda$ (where $P_\lambda = W_{\lambda} \ot V_{\lambda}$) in terms of the Weyl and Specht modules $W_{\lambda}$ and $V_{\lambda}$, respectively. We show the following statement, which describes the effect of the Haar twirl.

\begin{lemma}[Action of Haar twirl]\label{lem:haar-twirl} Let $\rho_{\lambda} = \dfrac{\ind_{W_\lambda}}{\Tr[\ind_{W_\lambda}]} $ be the maximally mixed state on $W_{\lambda}$.
Then
\begin{align}\label{eqn:haar-twirl}
    \     \E_{U \sim \Haar} (U^{\ot t})_{\reg A}  \proj{\phi}_{\reg{A \tilde E}} (U^{\ot t,\dagger})_{\reg A} = \sum_{\lambda \vdash t} \rho_{\lambda} \otimes \Tr_{W_\lambda}[\ind_{P_\lambda} \proj{\phi}_{\reg{A\tilde E}}\ind_{P_\lambda}]\,. 
\end{align}
\end{lemma}

We prove the above lemma in \Cref{sec:haar-twirl}. In fact, the same proof also shows that \Cref{lem:haar-twirl} holds for arbitrary input states, but for the proof of \Cref{lem:haar_pf_close} we only need to consider states supported on the distinct subspace.  

Next, recalling the decomposition of the distinct subspace projector given by \Cref{lem:perm_inv_sw}, we show that the action of the permutation-phase twirl results in the following; we defer the proof to \Cref{sec:pf-twirl}.

\begin{lemma}[Action of permutation-phase twirl]\label{lem:pf-twirl} Denoting by $\sigma_{\lambda} = \dfrac{\Lambda^{(\lambda)}_{W_\lambda}}{\Tr[\Lambda^{(\lambda)}_{W_\lambda}]} $ the maximally mixed state on the subspace $\supp(\Lambda^{(\lambda)}_{W_\lambda}) \cap W_{\lambda}$, we have 
    \[ \E_{P,F} (PF)_{\reg A}^{\ot t}  \proj{\phi}_{\reg{A \tilde E}}((PF)_{\reg A}^{\ot t})^\dagger = \sum_{\lambda \vdash t} \sigma_{\lambda} \otimes \Tr_{W_\lambda}[\ind_{P_\lambda} \proj{\phi}_{\reg{A \tilde E}}\ind_{P_\lambda}]\,.\]
\end{lemma} 

We now show how \Cref{lem:haar_pf_close} follows from \Cref{lem:haar-twirl,lem:pf-twirl}. Since all the subspaces $P_{\lambda}$ are orthogonal, we have 
\begin{align}\label{eqn:main1}
\norm{\momhaart{\proj{\phi}_{\reg{A \tilde E}}} - \mompft{\proj{\phi}_{\reg{A \tilde E}}}}_1 & = \sum_{\lambda} \norm{\left( \rho_{\lambda} - \sigma_{\lambda} \right) \ot \Tr_{W_\lambda}[\ind_{P_\lambda} \proj{\phi}_{\reg{A\tilde E}}\ind_{P_\lambda}]}_1 \notag \\
&= \sum_{\lambda} \norm{ \rho_{\lambda} - \sigma_{\lambda}}_1 \cdot \norm{\Tr_{W_\lambda}[\ind_{P_\lambda} \proj{\phi}_{\reg{A\tilde E}}\ind_{P_\lambda}]}_1 \notag \\
\ & \le \max_{\lambda} \norm{ \rho_{\lambda} - \sigma_{\lambda}}_1,
\end{align}
where the second equality used that the 1-norm (trace norm) is multiplicative under tensor products, while the inequality follows from the fact that $\sum_{\lambda} \|\Tr_{W_\lambda}[\ind_{P_\lambda} \proj{\phi}_{\reg{A\tilde E}}\ind_{P_\lambda}]\|=1$, which can be seen by taking the trace on both sides of \Cref{eqn:haar-twirl}.

Note that given a vector space $A = B \oplus B^\bot$, we have that 
\[ \norm{ \frac{\id_{A}}{\dim(A)} - \frac{\id_{B}}{\dim(B)}}_1 = \norm{\frac{\id_{B}}{\dim(A)} - \frac{\id_{B}}{\dim(B)}}_1 + \norm{\frac{\id_{B^\bot}}{\dim(A)}}_1 = 2-2\frac{\dim(B)}{\dim(A)}.\]

Since $\Lambda^{(\lambda)}_{W_\lambda}$ is a projector on a subspace of $W_{\lambda}$, we obtain the following for any $\lambda \vdash t$:
\begin{align}\label{eqn:main2}
    \norm{ \rho_{\lambda} - \sigma_{\lambda}}_1 &= \norm{ \frac{\id_{W_{\lambda}}}{\Tr[\id_{W_{\lambda}}]} - \frac{\Lambda^{(\lambda)}_{W_\lambda}}{\Tr[\Lambda^{(\lambda)}_{W_\lambda}]}}_1 = 2 - 2 \frac{\Tr[\Lambda^{(\lambda)}_{W_\lambda}]}{\Tr[ \id_{W_{\lambda}}]}\,.
\end{align}

We first compute the trace in the numerator of the right hand side.

\begin{claim}\label{clm:distinct-trace}
${\Tr[\Lambda^{(\lambda)}_{W_\lambda}]} = \dfrac{\dim(V_{\lambda})}{t!} \cdot{\Tr[\Lambda]}$.
\end{claim}
\begin{proof}
   \Cref{lem:perm_inv_sw} implies that $\Lambda = \sum_{\lambda \vdash t} \Lambda^{(\lambda)}_{W_{\lambda}} \otimes \id_{V_{\lambda}}$. Since $P_{\lambda} = W_{\lambda} \ot V_{\lambda}$, it follows that
    \begin{align}\label{eqn:dis-trace}
   \Tr[\Lambda^{(\lambda)}_{W_\lambda}] = \frac{1}{\dim V_\lambda} \Tr[\Lambda \id_{P_\lambda}]\,.
    \end{align}
   Plugging in the expression for the projector $\id_{P_{\lambda}}$ from \Cref{lem:P_characters}, we get that
    \begin{align*}
    \Tr[\Lambda \id_{P_\lambda}] = \frac{\dim(V_\lambda)}{t!} \sum_\pi \chi_\lambda(\pi^{-1}) \Tr[\Lambda R_\pi] = \frac{\dim(V_\lambda)^2}{t!} \Tr[\Lambda] \,, \label{eqn:lambda_P_ip}
    \end{align*}
    where we used the fact that $\Tr(R_\pi \Lambda) = 0$ unless $\pi = e$, and that $\chi_\lambda(e) = \Tr[\id_{V_{\lambda}}] = \dim(V_\lambda)$. 
    This yields the desired result after insertion into \Cref{eqn:dis-trace}.
\end{proof}

Given the above, we can now compute the quantity in \Cref{eqn:main2} by using the dimension bounds for Weyl and Specht modules.

\begin{claim}\label{cm:ratio} $1- \dfrac{\Tr[\Lambda^{(\lambda)}_{W_\lambda}]}{\Tr[ \id_{W_{\lambda}}]} \le O(t^2/d).$
\end{claim}
\begin{proof}
Using \Cref{clm:distinct-trace} and \Cref{lem:dim_weyl_specht}, we have  
\[ 1- \dfrac{\Tr[\Lambda^{(\lambda)}_{W_\lambda}]}{\Tr[ \id_{W_{\lambda}}]} = 1- \frac{\dim(V_{\lambda})\Tr[\Lambda]}{t! \dim(W_{\lambda})} = 1- \frac{\Tr[\Lambda]}{\Pi_{(i,j) \in \lambda} (d + j - i)}\,.\]

Note that $\Tr[\Lambda] = \frac{d!}{(d-t)!} \ge (d-t)^t$ and $\Pi_{(i,j) \in \lambda} (d + j - i) \le (d+t)^t$, since there are at most $t$ boxes in the Young diagram corresponding to $\lambda$, and the coordinates $i,j$ range from $1$ to $t$. Thus,
\[ 1- \dfrac{\Tr[\Lambda^{(\lambda)}_{W_\lambda}]}{\Tr[ \id_{W_{\lambda}}]} \le 1-  \left(\frac{d-t}{d+t}\right)^t \le 1 - \Big(\frac{1 - t/d}{1+t/d} \Big)^t \leq O(t^2/d)\,,\]
using that $t^2/d \ll 1$.
\end{proof}

Plugging the above bound into \Cref{eqn:main2} and  \Cref{eqn:main1} completes the proof of \Cref{lem:haar_pf_close}.

\subsubsection{Action of the $t$-wise Haar twirl (proof of \Cref{lem:haar-twirl})}
\label{sec:haar-twirl}

In order to derive the expression given by \Cref{lem:haar-twirl}, we first compute the result of applying the $t$-wise Haar twirl on Schur-Weyl basis states.

\begin{lemma} \label{lem:haar_moment_sw}
Let $\reg A \cong (\C^d)^{\ot t}$ and let $\ket{\alpha} = \ket{w_{\lambda, i}} \ot \ket{v_{\lambda, j}}$ and $\ket{\beta} = \ket{w_{\lambda', i'}} \ot \ket{v_{\lambda', j'}}$ be Schur-Weyl basis states on $\reg A$. 
Then 
\begin{align*}
\momhaart{\ketbra{\alpha}{\beta}} = \begin{cases} \frac{\id_{W_{\lambda}}}{\dim W_{\lambda}} \ot \ketbra{v_{\lambda, j}}{v_{\lambda, j'}}\, , & \text{ if $\lambda = \lambda'$ and $i = i'$} \\ 0 \, , & \text{otherwise} \end{cases}\,.
\end{align*}
\end{lemma}
\begin{proof}
By Schur-Weyl duality (\Cref{lem:schur-weyl}), we have that $U^{\ot t} = \sum_{\lambda_1} U^{(\lambda_1)} \otimes \id_{V_{\lambda_1}}$ where $U^{(\lambda_1)}$ only acts on $W_{\lambda_1}$. Thus, 
\begin{align*}
\momhaart{\ketbra{\alpha}{\beta}}  = \sum_{\lambda_1, \lambda_2} \left( \E_{U \sim \HaarMeasure} U^{(\lambda_1)} \ketbra{w_{\lambda, i}}{w_{\lambda', i'}} U^{(\lambda_2), \dagger} \right) \ot  \id_{V_{\lambda_1}}\ketbra{v_{\lambda, j}}{v_{\lambda', j'}}  \id_{V_{\lambda_2}} \,.
\end{align*}
We may assume that $\lambda = \lambda'$, since the above is zero otherwise. Then, we have that 
\begin{align*}
\momhaart{\ketbra{\alpha}{\beta}}  = \left( \E_{U \sim \HaarMeasure} U^{(\lambda)} \ketbra{w_{\lambda, i}}{w_{\lambda, i'}} U^{(\lambda), \dagger} \right) \ot \ketbra{v_{\lambda, j}}{v_{\lambda, j'}} \,.
\end{align*}
We claim that the expression inside the parentheses above is zero unless $i=i'$, in which case it is the maximally mixed state on the subspace $W_{\lambda}$. This follows from a standard fact in representation theory called Schur's Lemma (see~\cite{fulton2013representation}, Lemma 1.7), which says that if $(\mu, H)$ is an irreducible representations of a group $G$ and $T: H \to H$ is a linear map such that $T \circ \mu = \mu \circ T$ (such a map is called an \emph{intertwiner}), then $T = \gamma \cdot \id_{H}$ for some scalar $\gamma \in \C$. Applying it to our setting, we see that the operator
$$
T = \E_{ U \sim \HaarMeasure} {U}^{(\lambda)} \ketbra{w_{\lambda, i}}{w_{\lambda, i'}} {U}^{(\lambda),\dag} $$ is an intertwiner for the irrep $(U^{(\lambda)}, W_{\lambda})$. This fact is a simple consequence of Haar invariance, since
\begin{align*}
T \tilde{U}^{(\lambda)} &= \left(\E_{U \sim \HaarMeasure} {U}^{(\lambda)} \ketbra{w_{\lambda, i}}{w_{\lambda, i'}} {U}^{(\lambda),\dag} \right) \tilde{U}^{(\lambda)} \\
&=\E_{ U \sim \HaarMeasure} {U}^{(\lambda)} \ketbra{w_{\lambda, i}}{w_{\lambda, i'}} \left(\tilde{U}^{(\lambda),\dag} {U}^{(\lambda)}\right)^\dag\\
&=\E_{U \sim \HaarMeasure} \left(\tilde{U}^{(\lambda)}{U}^{(\lambda)}\right) \ketbra{w_{\lambda, i}}{w_{\lambda, i'}} {U}^{(\lambda),\dag} = \tilde{U}^{(\lambda)} T.
\end{align*}

Therefore, by Schur's Lemma, 
\begin{align*}
\E_{U \sim \HaarMeasure} U^{(\lambda)} \ketbra{w_{\lambda, i}}{w_{\lambda, i'}} U^{(\lambda),\dag} = \gamma_{\lambda, i, i'} \id_{W_{\lambda}}
\end{align*}
for some scalars $\alpha_{\lambda, i, i'}$. Since the operator on the left is traceless if $i \neq i'$, we have that $\alpha_{\lambda, i, i'} = 0$ unless $i = i'$. For $i = i'$, the normalisation follows because the operator on the left is a mixed state, i.e.~it has unit trace. This completes the proof.
\end{proof}

We can now compute the result of applying a $t$-wise Haar twirl to a general state.
\begin{proof}[Proof of \Cref{lem:haar-twirl}]
Expanding in the Schur-Weyl basis, 
\begin{align*}
\ket{\phi}_{\reg{A \tilde E}} = \sum_{\lambda, i, j} \left( \ket{w_{\lambda, i}} \ket{v_{\lambda, j}} \right)_{\reg A} \ot \ket{e_{\lambda, i, j}}_{\reg {\tilde E}}
\end{align*}
for not necessarily normalised vectors $\ket{e_{\lambda, i, j}}_{\reg {\tilde E}}$. It follows from \Cref{lem:haar_moment_sw} and linearity that 
\begin{align*}
\momhaart{\proj{\phi}_{\reg{A \tilde E}}} &= \sum_{\lambda, i, j, j'} \left( 
\frac{\id_{W_{\lambda}}}{\dim W_{\lambda}} \ot \ketbra{v_{\lambda, j}}{v_{\lambda, j'}} \right)_{\reg A} \ot \ketbra{e_{\lambda, i, j}}{e_{\lambda, i, j'}}_{\reg {\tilde E}} \,
\\ & = \sum_{\lambda \vdash t} \rho_{\lambda} \otimes \Tr_{W_\lambda}[\ind_{P_\lambda} \proj{\phi}_{\reg{A\tilde E}}\ind_{P_\lambda}],
\end{align*}
where $\rho_\lambda$ is the maximally mixed state on the subspace $W_{\lambda}$. 
\end{proof}

\subsubsection{Permutation-phase twirl on distinct subspace (proof of \Cref{lem:pf-twirl})}
\label{sec:pf-twirl}

In order to derive an exact expression for the permutation-phase twirl on the distinct subspace with projector $\Lambda$, we start with the following lemma.

\begin{lemma} \label{lem:dis_string_to_perm}
Let $\bx, \by \in \distinct(d,t)$. 
Then
\begin{align*}
\mompft{ \ketbra{\bx}{\by}} = \begin{cases} \dfrac{\Lambda R_{\sigma}}{\Tr[\Lambda]}  & \text{ if } \by = \bx_{\sigma} \text{ for } \sigma \in S_t,\\ 0 & \text{ otherwise.} \end{cases} 
\end{align*}
\end{lemma}
Note that since $\bx, \by \in \distinct(d,t)$, there exists at most one permutation for which $\by = \bx_{\sigma}$. 
\begin{proof}
Applying the $t$-wise $F$-twirl first, we get that
\begin{align*}
\E_{F} F^{\ot t} \ketbra{\bx}{\by} F^{\ot t, \dagger}
&= \left( \E_f (-1)^{\sum_i f(x_i) + f(y_i)}\right) \ketbra{\bx}{\by} \,.
\end{align*}
Here, the expectation $\E_f$ is over a uniformly random function $f: [d] \to \bit$.
Because $\bx$ and $\by$ are both tuples of distinct strings, it is easy to see that
\begin{align*}
 \E_f \left[(-1)^{\sum_i f(x_i) + f(y_i)} \right] = 
 \begin{cases}
1, & \text{ if } \by = \bx_{\sigma}, \text{ for some } \sigma \in S_t, \text{ and }\\
0, & \text{ otherwise}.
 \end{cases}
\end{align*}
Next applying the $t$-wise $P$-twirl, 
\begin{align*}
\E_{P} P^{\ot t}  \left(\E_{F} F^{\ot t} \ketbra{\bx}{\bx} F^{\ot t, \dagger}\right) P^{\ot t, \dagger} 
&= \E_{P} P^{\ot t} \ketbra{\bx}{\bx}R_{\sigma} P^{\ot t, \dagger} = \E_{P} P^{\ot t}\proj{\bx} P^{\ot t, \dagger} R_\sigma \,,
\end{align*}
where we used that  $P^{\ot t}$ commutes with $R_\sigma$.
To conclude, we note that for any tuple of distinct strings $\bx$, 
\begin{equation*}
\ \E_{P} P^{\ot t} \proj{\bx} P^{\ot t, \dagger} = \frac{\Lambda}{\Tr[\Lambda]}. \qedhere
\end{equation*}
\end{proof}

We will also need the following technical result, which follows from the Schur orthogonality relations (\cref{lem:orthogonality}).

\begin{lemma} \label{lem:K_simplified}
Let $\ket{\alpha} = \ket{w_{\lambda, i}} \ot \ket{v_{\lambda, j}}$ and $\ket{\beta} = \ket{w_{\lambda', i'}} \ot \ket{v_{\lambda', j'}}$ be Schur-Weyl basis states.
Then
\begin{align*}
\sum_{\sigma \in S_t} \bra{\beta} R^\dagger_{\sigma} \ket{\alpha} R_\sigma = \delta_{\lambda,\lambda'} \delta_{i,i'} \cdot \frac{t!}{\dim V_{\lambda}} \cdot (\id_{W_\lambda} \ot \ketbra{v_{\lambda,j}}{v_{\lambda, j'}})\,.
\end{align*}
\end{lemma}

\begin{proof}
We compute the matrix elements of $R^\dagger_\sigma$ in the Schur-Weyl basis. Recall that \Cref{lem:schur-weyl} implies that $R_{\sigma} = \sum_\lambda \id_{W_\lambda} \ot R^{(\lambda)}_\sigma$, where $R^{(\lambda)}_\sigma$ is the irreducible sub-representation of $R_\sigma$ on the Specht module $V_\lambda$. Thus, 
\begin{align}\label{eqn:perm-decomp}
(\bra{w_{\lambda', i'}}\bra{v_{\lambda', j'}}) R^\dagger_{\sigma} (\ket{w_{\lambda, i}} \ket{v_{\lambda, j}})
&= (\bra{w_{\lambda', i'}} \bra{v_{\lambda', j'}}) \left( \sum_\lambda \id_{W_\lambda} \ot R^{(\lambda), \dagger}_\sigma \right) (\ket{w_{\lambda, i}} \ket{v_{\lambda, j}}) \notag \\
&= \delta_{\lambda,\lambda'} \delta_{i,i'}  \cdot \bra{v_{\lambda, j'}} R^{(\lambda), \dagger}_\sigma \ket{v_{\lambda, j}} \notag\\
\ & =  \delta_{\lambda,\lambda'} \delta_{i,i'} \cdot \overline{\bra{v_{\lambda, j}} R^{(\lambda)}_\sigma \ket{v_{\lambda, j'}}} \,.
\end{align}
Therefore, for the rest of the proof, we consider $\ket{\alpha}$ and $\ket{\beta}$ with $\lambda = \lambda'$ and $i = i'$.

Again using the decomposition of $R_{\sigma}$ in terms of its irreducible sub-representations together with \Cref{eqn:perm-decomp}, we can write 
\begin{align*}
    \sum_{\sigma \in S_t} (\bra{w_{\lambda, i}}\bra{v_{\lambda, j'}}) R^\dagger_{\sigma} (\ket{w_{\lambda, i}} \ket{v_{\lambda, j}}) R_{\sigma} = \sum_{\lambda_1 \vdash t}  \id_{W_{\lambda_1}} \ot \left(\sum_{\sigma \in S_t} \overline{\bra{v_{\lambda, j}} R^{(\lambda)}_\sigma \ket{v_{\lambda, j'}}} R^{(\lambda_1)}_\sigma \right)
\end{align*}

Schur's orthogonality relations (\Cref{lem:orthogonality}) now imply that the operator in the parentheses is zero unless $\lambda_1=\lambda$, in which case it equals 
\[ \sum_{\sigma \in S_t} \overline{\bra{v_{\lambda, j}} R^{(\lambda)}_\sigma \ket{v_{\lambda, j'}}} R^{(\lambda)}_\sigma  = \frac{t!}{\dim(V_{\lambda})} \ketbra{v_{\lambda,j}}{v_{\lambda,j'}}.\]
Plugging this in gives the desired result.
\end{proof}

We are now in a position to prove \Cref{lem:pf-twirl} which expresses the result of applying the permutation-phase twirl to any state that is only supported on distinct strings on $\reg A$ in terms of the Schur-Weyl subspaces. 

\begin{proof}[Proof of \Cref{lem:pf-twirl}] We first expand $\ket{\phi}$ in the standard basis on $\reg{A}$:
\begin{align}
\ket{\phi}_{\reg{A \tilde E}} = \sum_{\substack{\bx  \in \distinct(d,t)}} \ket{\bx}_{\reg A} \ket{\tilde e_{\bx}}_{\reg{ \tilde E}} \,, \label{eqn:std_basis_expansion}
\end{align}
where $\ket{\tilde e_{\bx}}_{\reg E}$'s are unnormalized and not necessarily orthogonal vectors and we used that $\ket{\phi}_{\reg{A \tilde E}}$ is supported over the distinct subspace in the register $A$. 

Applying the permutation-phase twirl to the register $\reg{A}$, we have by linearity, 
\begin{align*}
\mompft{\proj{\phi}_{\reg{AE}}} &= 
\sum_{\substack{\bx,\by \in \distinct(d,t)}} \left( \mompft{\ketbra{\bx}{\by}_{\reg  A}} \right) \ot \ketbra{\tilde e_{\bx}}{\tilde e_{\by}}_{\reg{\tilde E}} \,.
\end{align*}
\Cref{lem:dis_string_to_perm} implies that the term in the parentheses is non-zero only when $\by = \bx_{\sigma}$ for some $\sigma \in S_t$. Since we sum over all possible $\bx, \by \in \distinct(d,t)$ and there is at most one such $\sigma$ for each pair of tuples $\bx$ and $\by$, it follows that 
\begin{align}\label{eqn:pf1}
\mompft{\proj{\phi}_{\reg{AE}}} 
&= \sum_{\substack{\bx \in \distinct(d,t)\\ \sigma \in S_t}}  \frac{\left(\Lambda R_\sigma\right)_{\reg A}}{\Tr[\Lambda]}  \ot \ketbra{\tilde e_{\bx}}{\tilde e_{{\bx_{\sigma}}}}_{\reg{\tilde E}} \notag \\
&= \frac{\Lambda_{\reg A}}{\Tr[\Lambda]} \sum_{\sigma \in S_t}  (R_\sigma)_{\reg A} \ot \Bigg( 
\sum_{\bx \in \distinct(d,t)}  \ketbra{\tilde e_{\bx}}{\tilde e_{{\bx_{\sigma}}}}_{\reg{\tilde E}}  \Bigg) \notag \\
&= \frac{\Lambda_{\reg A}}{\Tr[\Lambda]} \sum_{\sigma \in S_t}  (R_\sigma)_{\reg A} \ot \Tr_{\reg A}\left[(R^\dagger_{\sigma} \ot \id_{\reg{\tilde E}}) \proj{\phi}_{\reg{A \tilde E}}\right]\,.
\end{align}
For the second equality we simply rearranged sums, and for the last equality we wrote the expression in parentheses in the second line more compactly as a partial trace, which follows directly from the expansion in \Cref{eqn:std_basis_expansion}.

We now rewrite the above partial trace in the Schur-Weyl basis on $\reg A$. Let 
\begin{align}
    \ket{\phi}_{\reg{A \tilde E}} = \sum_{\lambda,i,j} (\ket{w_{\lambda, i}} \ket{v_{\lambda, j}})_{\reg{A}} \ket{e_{\lambda,i,j}}_{\reg{ \tilde E}} \label{eqn:phi_schur_expansion}
\end{align}
be the state in the Schur-Weyl basis, where $\ket{e_{\lambda,i,j}}_{\reg{ \tilde E}}$ are unnormalized and not necessarily orthogonal vectors.
Then, 
\begin{align*}
\Tr_{\reg A}\left[(R^\dagger_{\sigma} \ot \id_{\reg{\tilde E}}) \proj{\phi}_{\reg{A \tilde E}}\right] =  \sum_{\substack{\lambda, i, j \\ \lambda', i', j'}} (\bra{w_{\lambda', i'}}\bra{v_{\lambda', j'}}) R^\dagger_{\sigma} (\ket{w_{\lambda, i}} \ket{v_{\lambda, j}}) \ketbra{e_{\lambda, i, j}}{e_{\lambda', i', j'}}_{\reg{\tilde E}}.
\end{align*}
Plugging the above into \Cref{eqn:pf1}, 
\begin{align*}
\mompft{\proj{\phi}_{\reg{A \tilde E}}}
&= \frac{\Lambda_{\reg A}}{\Tr[\Lambda]} \sum_{\substack{\lambda, i, j \\ \lambda', i', j'}} \left( \sum_{\sigma} (\bra{w_{\lambda', i'}}\bra{v_{\lambda', j'}}) R^\dagger_{\sigma} (\ket{w_{\lambda, i}} \ket{v_{\lambda, j}}) R_\sigma \right)_{\reg A} \ot \ketbra{e_{\lambda, i, j}}{e_{\lambda', i', j'}}_{\reg{\tilde E}} \,.
\end{align*}
Applying \Cref{lem:K_simplified} to the term in parentheses, we can simplify this to 
\begin{align}\label{eqn:pf2}
\mompft{\proj{\phi}_{\reg{A \tilde E}}}
&= \frac{\Lambda_{\reg A}}{\Tr[\Lambda]} \sum_{\substack{\lambda, i, j, j'}} \frac{t!}{\dim (V_{\lambda})} \left(\id_{W_\lambda} \ot \ketbra{v_{\lambda,j}}{v_{\lambda, j'}}\right)_{\reg A} \ot \ketbra{e_{\lambda, i, j}}{e_{\lambda, i, j'}}_{\reg {\tilde E}}\,.
\end{align}
By \Cref{lem:perm_inv_sw}, we can write $\Lambda = \bigoplus_{\lambda \vdash t} \Lambda^{(\lambda)}_{W_{\lambda}} \ot \id_{V_{\lambda}}$, where $\Lambda^{(\lambda)}_{W_\lambda}$ are projectors supported on $W_{\lambda}$. Since $\id_{W_\lambda}$ is simply the identity on subspace $W_\lambda$, this implies 
\begin{align*}
\Lambda (\id_{W_\lambda} \ot \ketbra{v_{\lambda,j}}{v_{\lambda, j'}}) = \Lambda^{(\lambda)}_{W_\lambda} \ot \ketbra{v_{\lambda,j}}{v_{\lambda, j'}} \,.
\end{align*}
Plugging this into \Cref{eqn:pf2} and rewriting it as a partial trace,
\begin{align*}
\mompft{\proj{\phi}_{\reg{A \tilde E}}}
    & = \sum_{\lambda \vdash t}  \dfrac{t!}{\dim(V_{\lambda}) \Tr[\Lambda]} \Lambda^{(\lambda)}_{W_\lambda} \otimes \Tr_{W_\lambda}[\ind_{P_\lambda} \proj{\phi}_{\reg{A \tilde E}}\ind_{P_\lambda}] \,.
\end{align*}
Since $\Tr[\Lambda^{(\lambda)}_{W_\lambda}] = (\dim(V_{\lambda}) \Tr[\Lambda])/t!$ as we showed in \Cref{clm:distinct-trace}, the first tensor factor is indeed the maximally mixed state on the support of the projector $\Lambda^{(\lambda)}_{W_\lambda}$. This completes the proof.
\end{proof}

\bibliographystyle{alpha}
\bibliography{main}

\end{document}